\documentclass[conference,10pt]{IEEEtran}
\usepackage{cite}
\usepackage{times}
\usepackage{amsmath}
\usepackage{bbm}
\usepackage{mathrsfs}
\usepackage{amsbsy}
\usepackage{amsxtra}
\usepackage{amsopn}
\usepackage{latexsym} 
\usepackage{amsfonts}
\usepackage{epsfig}
\usepackage{graphicx}
\usepackage{amssymb}
\usepackage{stmaryrd}
\usepackage{algorithm}
\usepackage{algorithmic}
\usepackage{color,soul}
\usepackage{epstopdf}
\usepackage{amsthm}
\usepackage{capt-of}
\newtheorem{theorem}{Theorem}

\usepackage{gensymb}
\usepackage{lipsum}

\usepackage{url}
\renewcommand{\hl}[1]{#1}
\usepackage{booktabs} 
\DeclareMathOperator*{\argmax}{arg\,max}

\usepackage{array}
\newcolumntype{C}[1]{>{\centering\arraybackslash}m{#1}}
\usepackage{enumitem}
\usepackage[font=footnotesize]{caption}

\title{Crowd Counting Through Walls Using WiFi \vspace{-0.12in}}
\author{Saandeep Depatla and Yasamin Mostofi}
\begin{document}
\maketitle
\begin{abstract}
Counting the number of people inside a building, from outside and without entering the building, is crucial for many applications. In this paper, we are interested in counting the total number of people walking inside a building (or in general behind walls), using readily-deployable WiFi transceivers that are installed outside the building, and only based on WiFi RSSI measurements. The key observation of the paper is that the inter-event times, corresponding to the dip events of the received signal, are \hl{fairly} robust to the attenuation through walls (for instance as compared to the exact dip values). We then propose a methodology that can extract the total number of people from the inter-event times. More specifically, we first show how to characterize the wireless received power measurements as a superposition of renewal-type processes. By borrowing theories from the renewal-process literature, we then show how the probability mass function of the inter-event times carries vital information on the number of people. We validate our framework with $44$ experiments in five different areas on our campus ($3$ classrooms, a conference room, and a hallway), using only one WiFi transmitter and receiver installed outside of the building, and for up to and including $20$ people. Our experiments further include areas with different wall materials, such as concrete, plaster, and wood, to validate the robustness of the proposed approach. Overall, our results show that our approach can estimate the total number of people behind the walls with a high accuracy while minimizing the need for prior calibrations.
\end{abstract}
\vspace{-0.05in}
\begin{section}{Introduction}{\label{sec_intro}}
The ability to estimate the total number of people in an area can be useful for several applications. For instance, smart buildings can optimize the energy consumption based on the number of people in the building \cite{agarwal2010occupancy,ardakanian2016non}. Retails can better plan their business by assessing which parts of the store get more visitors \cite{perdikaki2012effect}. Smart cities can better plan the resources by estimating which areas of the city are more crowded \cite{BlueScan}.

A survey of the literature indicates that the problem of crowd counting has been investigated by researchers from computer-vision, wireless networking, and environmental science communities. In computer-vision, for instance, photographic images of an area are used to identify the number of people present in the area \cite{nichols2008multi,li2008estimating,lin2001estimation}. However, these methods 1) require a network of cameras to be installed in the area of interest and as such have a high deployment cost, 2) cannot work in the dark, 3) cannot work behind walls, and 4) pose privacy issues. Researchers in the environmental science community utilize the characteristics of the area of interest such as temperature, concentration of carbon dioxide, and dew point to identify the number of people in the area \cite{lam2009occupancy,jiang2016indoor,wang2000experimental,chenpredictive}. However, sensing the environment in this manner requires a direct access to the area of interest and cannot be used in areas occluded by walls or in areas where access is restricted. Furthermore, they require installing specialized sensors.

The ability of radio frequency (RF) signals to penetrate through objects, such as walls, combined with the ubiquity of wireless devices, such as WiFi routers, provide a great potential for imaging \cite{depatla2015x,gonzalez2014integrated,karanam20173d}, tracking \cite{wilson2009through}, and occupancy estimation using RF signals. Crowd counting based on wireless devices can be mainly classified into (i) device-based active and (ii) device-free passive methods. The device-based active methods rely on people to carry a communication device \cite{weppner2013bluetooth,wirz2013probing}, which can limit their applicability. For this reason, there has recently been a considerable interest in device-free methods, which do not require people to carry any device. Instead, device-free methods rely on the interaction of the wireless signals with the people in the area of interest.

In this context of device-free counting,  \cite{yuan2011crowd} classifies the crowd density in an area into low, medium, and high using a network of wireless nodes. \cite{yoshidaestimating} uses the variance of the WiFi received signal strength indicator (RSSI) to estimate up to $7$ people. In this approach, an extensive prior learning phase with different number of people is used. Furthermore, the approach requires a large number of wireless nodes ($10$ Rx and $1$ Tx). \cite{xu2013scpl} simultaneously estimates the number and the location of \hl{up to $4$ people with $22$} wireless nodes. \cite{di2016trained} uses  differential channel state information (CSI) to classify the number of people. The method has an extensive calibration phase and is only tested with up to $7$ people. \cite{xi2014electronic} counts up to $30$ people, using CSI measurements at $30$ subcarriers \hl{and with $4$ WiFi links} located in the area. The method requires an extensive training phase with $7$ experiments and up to $7$ people walking in the \hl{same area a priori.} \cite{depatla2015occupancy} counts up to $10$ people in an area using only the RSSI measurements of a single WiFi link by deriving a probability density function (PDF) of the received signal strength. 
While the approach in \cite{depatla2015occupancy} does not require extensive prior calibrations, such as having different number of people walk in the area, there is still a need to make measurements when a small number of people stand on the Line of Sight (LOS) link a priori. Furthermore, \cite{depatla2015occupancy} relies on labeling the dips of the received signal, which can be prone to errors in behind-wall scenarios due to the high attenuation by the walls.

In summary, great progress has been made towards crowd counting with WiFi signals. 
However, all the aforementioned work are on counting in the same room where the transceivers are located and do not count through walls. In other words, to the best of our knowledge, \textit{\hl{there is no work in the literature that has demonstrated through-wall counting}}. Furthermore, utilizing the existing work for the through-wall scenarios does not work. For instance, we tested \cite{xi2014electronic} \hl{in our through-wall settings and observed errors of up to $7$ people when $10$ people were present. Our previous approach} \cite{depatla2015occupancy}, \hl{on the other hand, relies on the dip values which can be highly attenuated and thus prone to measurement errors in through-wall settings. In summary, through-wall counting is a considerably challenging problem as the walls can heavily attenuate the signal, making the corresponding estimation problem more challenging,} which is the main motivation for this paper.
Moreover, most existing work on non-through wall counting have a demanding calibration phase that can be as involved as the main experiments.  For instance, \cite{xi2014electronic}, \cite{yoshidaestimating}, and \cite{di2016trained} require running multiple experiments where up to $7$ people walk in the area of interest. However, extensive calibration may not be feasible in through wall scenarios. In this paper, our proposed approach has a small calibration phase that does not have to be in the same environment. Finally, even when counting in non-behind wall settings, several existing work require a number of links for counting. For instance, \cite{xi2014electronic} and \cite{yoshidaestimating} have utilized $4$ and $11$ links respectively in areas with a comparable size to ours. In this paper, we show how to count up to $20$ people from behind walls with only one link.  Furthermore, we only utilize RSSI measurements for counting, which can be easily measured in any WiFi card, or can be implemented on any general wireless transceiver such as a Bluetooth device.

We next summarize our key contributions:
\begin{itemize}[leftmargin=*]
\vspace{-0.025in}
\item We show that the effect of a single person on the WiFi link can be modeled using a process that we refer to as a ``Renewal-type" random process.
\item We then show that the inter-event times carry vital information on the total number of people, and are more robust to the attenuation caused by the walls (as compared to the dip values), enabling a high-accuracy estimation through walls. More specifically, we use theories from Renewal process literature to model the effect of $N$ people as a superposition of ``Renewal-type" processes. We then derive the Probability Mass Function (PMF) of the inter-event times based on this model, and use it to estimate the number of people using a maximum likelihood (ML) estimator. \hl{It is noteworthy that no existing work has shown the relationship between inter-event times and the total number of people.}

\item We extensively validate our framework using $44$ real experiments in five different areas on our campus, three classrooms, a conference room, and a hallway (see Fig. \ref{fig_classroom_scenario}, \ref{fig_hfh_wall}, \ref{fig_arts}, \ref{fig_wooden_trailer}, and \ref{fig_classroom_phelps_20ppl}).  More specifically, we show that we can estimate up to and including $20$ people with an error of $2$ people or less $100$\% of the time and with an error of $1$ person or less $75$\% of the time. Our experiments further include areas with different wall materials, such as concrete, plaster, and wood, to validate the robustness of our approach.
\end{itemize}

The rest of the paper is organized as following. In Section \ref{sec_prob_statement}, we summarize our motion model and discuss the impact of the movement of the people on wireless channel measurements. In Section \ref{sec_estimation_method}, we then propose our framework to estimate the total number of people by using properties of the inter-event times. 
In Section \ref{sec_exp_results}, we thoroughly validate our framework using several experiments in five different areas on our campus. 
We conclude in Section \ref{sec_conclusions}.
\end{section}
\begin{section}{Problem Setup }{\label{sec_prob_statement}}
\begin{figure}[t]
\centering
\includegraphics[width=0.53\linewidth]{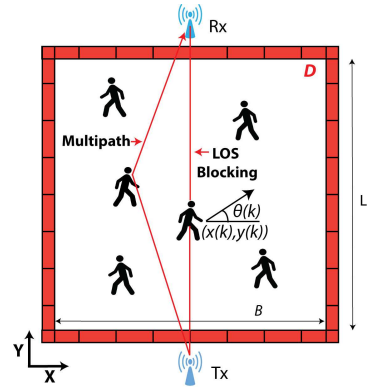}
\vspace{-0.05in}
\caption{An  illustration of the workspace with people walking inside. The red outer boundary denotes the walls. The WiFi Tx and Rx are located behind the walls and collect wireless measurements as people walk in the region. The goal of this paper is then to estimate the number of people in the workspace using only the wireless measurements. People affect the link in two major ways: LOS blockage and multipath, as shown.}\label{fig_workspace}
\vspace{-0.25in}
\end{figure}
Consider an area that is enclosed by walls, such as a room, where $N$ people are walking. Fig. \ref{fig_workspace} shows an example of this. Our goal is then to estimate the number of people walking in this area,  using only RSSI measurements of WiFi nodes that are located outside of the area. In this section, we first summarize the motion model of people and then briefly discuss the impact of movement of people on wireless measurements.
\vspace{-0.21in}
\begin{subsection}{Motion Model}\label{sec_motion_model}
In this paper, we assume that people are walking casually in the area of interest. In \cite{depatla2015occupancy}, we proposed a simple motion model to model the casual motion of people. In this paper, we adopt this model in our theoretical derivations. We next briefly summarize this motion model.

Consider the motion of a single person in the workspace $D$ of Fig. \ref{fig_workspace}. Let $x(k)$ and $y(k)$ denote the position of the person along x and y-axis respectively at time instant $k$. Furthermore, let $\theta(k)$ represent the heading of the motion w.r.t the x-axis and at time instant $k$, as indicated in Fig.\  \ref{fig_workspace}.\footnote{Throughout this paper, we use time instant $k$ to refer to $k\delta t$ for notational convenience, where $\delta t$ is the time step.} For the sake of mathematical simplicity, $x(k),\ y(k)$, and $\theta(k)$ are assumed to only take discrete values. The following model then captures a casual walk \cite{depatla2015occupancy}:
\begin{equation} \label{eq_head_model}
\theta(k) =
\begin{cases}
\theta(k-1) &\textnormal{with probability\ } p \\
\textnormal{Uniformly from \ } \mu  &\textnormal{with probability\ } 1-p,
\end{cases}	
\end{equation}
where $\mu =[0, \Delta\theta, 2\Delta\theta,\dots ,2\pi-\Delta\theta]$, and $\Delta\theta$ is the discretization step size for the heading angle.  Given the heading direction, the position dynamics is then given by the following:
\begin{equation}\label{eq_position_model}
\begin{split}
x(k+1) &= x(k) + v \delta t\  \textnormal{cos}(\theta(k)) \\
y(k+1) &= y(k) + v \delta t\  \textnormal{sin}(\theta(k)),
\end{split}
\end{equation}
where $v$ is the speed of the person, and $\delta t $ is the time step. Equation (\ref{eq_head_model}) basically describes a casual walk in which a person maintains her/his direction for a while, after which she/he may switch to a new direction. When a person reaches the boundary of $D$, she/he is furthermore assumed to reflect off of the boundary, similar to a ray of light. Equation (\ref{eq_position_model}) can then be extended to account for the boundary behavior. We refer the readers to \cite{depatla2015occupancy} for the corresponding expression for the sake of brevity.  Based on this simple dynamics, the motion can be modeled as a discrete-time Markov chain, as shown in \cite{depatla2015occupancy}.  We adopt this motion model in this paper and use the properties of the underlying Markov chain in the subsequent sections.  We note that during our experiments, we simply ask people to walk casually, without any specific instruction, and the aforementioned model is only used for the purpose of mathematical derivations.
\end{subsection}
\vspace{-0.055in}
\subsection{Impact of People on Wireless Measurements}\label{subsec_effect_ppl_RSS}
\vspace{-0.01in}
Consider the scenario shown in Fig. \ref{fig_workspace}, where multiple people are walking inside a building. A pair of WiFi nodes are located outside of the building. The WiFi transmitter (Tx) transmits wireless signals that interact with the walking people and the static objects in the area of interest, and are then received by the receiver (Rx). In general, properly capturing the interaction of the people with the transmitted signal requires detailed wave modeling to capture several propagation phenomena. In \cite{depatla2015occupancy}, it was shown that the two main phenomena of LOS blockage and multipath suffice to capture the impact of walking people on a wireless transmission. We next briefly summarize these two impacts: \\
(i)\textit{LOS Blocking:} When $l>0$ people are along the line joining the Tx and Rx (i.e., the LOS link), the received power measurements are significantly attenuated. \\
(ii)\textit{Multipath Effect:} The wireless signals from the Tx get reflected off of the people and interfere constructively/destructively at the Rx, depending on the position of the people. This causes the wireless measurements to fluctuate as people are walking. 

Fig. \ref{fig_workspace} illustrates the LOS blocking and multipath effects. The fluctuations of the received signal power, due to both LOS blocking and multipath effects, implicitly contain information about the total number of people walking in $D$ \cite{depatla2015occupancy}. In the next section, we propose a method for crowd counting behind the walls, based on LOS blockage and Renewal theory.
\end{section}
\vspace{-0.05in}
\begin{section}{Crowd Counting Behind the Walls}{\label{sec_estimation_method}}
In this section, we propose a new framework to estimate the number of people walking inside an occluded area using only the RSSI of WiFi nodes located outside of the area. Specifically, we first model the motion of a single person as a discrete-time random process. We then utilize theories from the Renewal process literature to characterize the impact of multiple people and identify the statistics that can be used to estimate the number of people. As we shall see, the inter-event times of the resulting process carry vital information on the number of people, as we shall mathematically characterize.	
\vspace{-0.075in}
\begin{subsection}{Motion of a single person as a Renewal-type process}\label{sec_motion_single_person}
Consider a scenario where $N$ people are walking in the workspace $D$, as shown in Fig. \ref{fig_workspace}. Without loss of generality, assume that the transmitter and the receiver are at the midpoint of the corresponding walls. We then say a person crosses the LOS link at time $k+1$, if $x(k+1)  \geq \frac{B}{2} \textnormal{\ and\ } x(k) \leq \frac{B}{2} $ or $ x(k+1)  \leq \frac{B}{2} \textnormal{\ and\ } x(k) \geq \frac{B}{2} $.
The time intervals between successive crosses (i.e., LOS crossings) implicitly carry vital information about the total number of people walking in the area, as we show in the next section. In this section, we first mathematically characterize the statistics of the time intervals between successive crosses, when a single person is walking in $D$. We then utilize the results derived here to model the impact of $N$ on the statistics of the cross times in the next section, when $N$ people are walking in $D$.
		
Consider a single person walking in the workspace $D$. Let $E$ denote an event of the person crossing the LOS link. Due to the non-deterministic nature of the walk, the times at which event $E$ happens are random in nature. Let $X_1, X_2, \dots ,X_T$ denote a sequence of random variables such that, 
\begin{equation} \label{eq_single_event_seq}
X_i= \begin{cases}
1 & \textnormal{if\ } E \textnormal{\ happens at time instant\ } i \\
0 & \textnormal{otherwise}.
\end{cases}
\end{equation}
Let $S_1, S_2, \dots , S_{n+1}$ denote the times at which event $E$ occurs and let $T_1, T_2, \dots ,T_{n}$ denote the inter-event times.  As mentioned in Section \ref{sec_motion_model}, we have discretized the time to a step size of $\delta t$. Thus, $S_i$, for $1 \leq i\leq n+1$, and $T_i$, for $1 \leq i \leq n$, are non-negative integers.

Fig. \ref{fig_single_renewal_process} shows a sample realization of the process along with the sample occurrence times and the inter-event times. Under the casual motion model of (\ref{eq_head_model}), the positions/headings will have a uniform distribution asymptotically in theory and after a sufficient time in practice \cite{depatla2015occupancy}. We thus assume that the positions/headings have no spatial bias in our derivations.
\begin{figure}[t]
\centering
\includegraphics[width=0.95\linewidth]{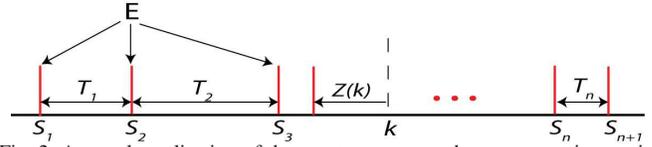}
\vspace{-0.1in}
\caption{A sample realization of the event sequence, where an event is crossing the LOS link. The events occur at $S_1, S_2, \dots , S_{n+1}$. The inter-event times are denoted by $T_1, T_2, \dots ,T_{n}$.  }\label{fig_single_renewal_process}
\vspace{-0.3in}
\end{figure}
Then, we have,
\begin{equation}
\begin{split}
P(T_i=k) & = P(X_{i+k}=1, \dots X_{i+1}=0\  |\  X_i=1) \\
& =  P(X_{j+k}=1,\dots X_{j+1}=0\ |\  X_j=1) \\
& = P(T_j=k)\  \forall\  i,j, \textnormal{and\ } k,
\end{split}
\end{equation}
where $P(.)$ denotes the probability of the argument. This implies that $\left \{T_i \right \}, \textnormal{for} \  i \in \left \{1,2, \dots , n\right \}$, are identically distributed. If the inter-event times are also independent, then the process is called a Renewal process \cite{barbu2009semi}. However, the inter-event times of our case are not necessarily independent. We thus use the term ``Renewal-type process" in this paper, to refer to this type of process where the inter-event times are identically distributed but not independent. We next characterize the PMF of the inter-event times.

Let $f$ denote the PMF of the inter-event times $T_i$. Let $Z(k)$ denote the backward recurrence time at $k$, i.e., the time from time instant $k$ that we need to travel back before encountering an event, as shown in Fig. \ref{fig_single_renewal_process}. Let $g(z;k)$ denote the PMF of $Z(k)$. We next characterize the relationship between $f$ and $g(z;k)$, which we shall utilize in Section \ref{sec_multiple_people}. 

Let $h(k)$ denote the probability that $E$ occurs at time $k$, i.e.,
$h(k)= P(k=S_j) \textnormal{\ for some\ } j$, 
where $P(.)$ denotes the probability of the argument. Then, $g(z;k)$, i.e., the probability that we need to travel backward $z$ time steps from time $k$ to encounter an event, is the product of the probability of an event occurring at time $k-z$ and the probability that there is no event at times $\{k-z+1, k-z+2, \dots k-1 \}$, given that an event occurs at $k-z$. Formally, $g(z;k)$ can be written as
\begin{equation}\label{eq_relation_delay_pmf}
g(z;k)= h(k-z)F_c(z),
\end{equation}
where $F_c(z)$ is the complimentary cumulative distribution function (CCDF) of the inter-event times.  As shown in \cite{depatla2015occupancy}, $h$ is given by the following:
\begin{equation}\label{eq_aysmpt_pc}
p_c\triangleq  h(k-z) = \frac{2v\delta t}{B\pi},\  \forall\  k \geq z.
\end{equation}
Therefore,
\begin{equation}\label{eq_delay_pmf_asympt}
g(z;k) = p_c F_c(z), \  \forall \ k \geq z.
\end{equation}
\vspace{-0.03in}		
\end{subsection}	
\vspace{-0.25in}
\begin{subsection}{Motion of Multiple People as a Superposition of Renewal-type Processes} \label{sec_multiple_people}
\vspace{-0.03in}
In this section, we characterize the PMF of the inter-event times when $N$ people are walking in $D$ and show that it contains useful information about the total number of people $N$. We then propose a ML estimator to estimate $N$, based on our characterization of the inter-event times.

Consider $N$ people walking in the workspace $D$. Let $\{X_i^j\}$, for $\ 1 \leq i \leq T$, denote the sequence of events as defined in (\ref{eq_single_event_seq}), but for the $j^\textnormal{th}$ person. Let $\{Y_i\}$, for $\  1 \leq i \leq T$, denote the corresponding superposed sequence. We define $\{Y_i\}$ as $Y_i = {\displaystyle \sum_{j=1}^{N} X_i^j}$.
Fig. \ref{fig_superposed_renewal_process} shows sample individual and superposed event sequences, for the case of $N$ people, along with their sample occurrence and inter-event times. For the superposed sequence of events, we say that an event occurred at time $i$ if $Y_i \neq 0$. In other words, an event occurs at time $i$ if at least one person crosses the LOS link at time $i$. Since multiple events can occur at the same time, we have $Y_i \in \{0,1,...,N\}$. However, we do not distinguish the events based on the value of $Y_i$, as our proposed method does not rely on the exact values of $Y_i$ and only depends on if it is zero or non-zero, which will result in a more robust estimator to measurement errors.

Let $f_p(z_p;N)$ denote the PMF of the inter-event times of the superposed process due to $N$ people. Let $Z_p(k)$ and $g_p(z_p;k)$ denote the backward recurrence time at $k$ and its corresponding PMF respectively.   
\begin{figure}[t]
\centering
\includegraphics[width=0.85\linewidth]{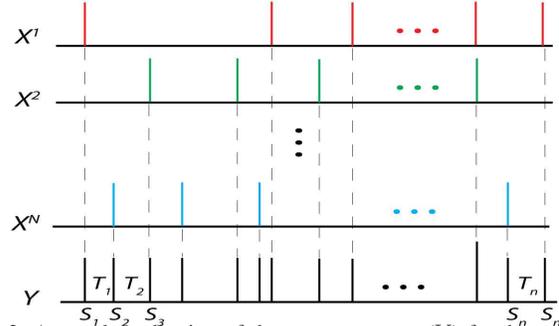}
\vspace{-0.12in}
\caption{A sample realization of the event sequence ($Y$) for the superposed process, which corresponds to $N$ people walking in the area of interest. An event $E$ here corresponds to any  crossing of the LOS link. The events occur at $S_1, S_2, \dots ,S_{n+1}$. The inter-event times are denoted by $T_1, T_2, \dots ,T_{n}$. The processes corresponding to individual people are also shown ($X_i$ s).}\label{fig_superposed_renewal_process}
\vspace{-0.15in}
\end{figure}
\vspace{-0.07in}   	
\begin{theorem}
We have the following expression for the PMF of the inter-event time: $f_p(z_p;N)= c \Delta g_p(z_p;k), \forall\  k \geq z_p $, where $c$ is a normalizing constant that is not a function of $N$, and $\Delta$ is the forward difference operator.
\end{theorem}
\begin{proof}
The backward recurrence time, $Z_p(k)$, for the superposed process can be written as
\begin{equation}
Z_p(k) = \textnormal{min} \left \{Z^1(k), Z^2(k), \dots , Z^N(k)\right \},
\end{equation}
where $Z^j(k), \text{for} \ 1 \leq j \leq N$, is the backward recurrence time for the $j^\textnormal{th}$ event sequence, and $\textnormal{min} \left \{ . \right \}$ denotes the minimum of the arguments. Then, since people are walking independently, we have,
\begin{equation}\label{eq_pooled_delay}
P(Z_p(k) \geq z_p) =  {\displaystyle \prod_{j=1}^{N} P(Z^j(k) \geq z_p)}.
\end{equation}
By substituting (\ref{eq_delay_pmf_asympt}) in (\ref{eq_pooled_delay}), we get,
\begin{equation}\label{eq_ccdf_pooled}
P(Z_p(k) \geq z_p)= {\displaystyle \Bigg[\sum_{m=z_p}^{\infty} p_c F_c(m)\Bigg]^N },\  \forall\  k \geq z_p,
\end{equation}
where $F_c(.)$ is the CCDF of the inter-event times for the case of $N=1$, and $p_c$ is the probability of crossing for the case of $N=1$, as defined in Section \ref{sec_motion_single_person}. From  (\ref{eq_ccdf_pooled}), we get the corresponding PMF as follows: 
\begin{equation}
g_p(z_p;k) = -\Delta P(Z_p(k) \geq z_p),\  \forall\  k \geq z_p.
\end{equation}
By following steps similar to (\ref{eq_relation_delay_pmf}), (\ref{eq_aysmpt_pc}), and (\ref{eq_delay_pmf_asympt}), we get the PMF of the inter-event times for the superposed process as follows, 
\begin{equation}\label{eq_PMF_N_people}
\begin{split}
&f_p(z_p;N) = c \Delta g_p(z_p;k) \textnormal{\ \ \ \ for\ } k \geq z_p\\
= c \Delta&\Bigg[{\displaystyle \Bigg(\sum_{m=z_p}^{\infty} p_c F_c(m)\Bigg)^N } - {\displaystyle \Bigg(\sum_{m=z_p+1}^{\infty} p_c F_c(m)\Bigg)^N }\Bigg].
\end{split}
\end{equation}
This proves the theorem. \hspace{2cm}
\end{proof}
It can be seen from (\ref{eq_PMF_N_people}) that the PMF of the inter-event times is an implicit function of the number of people $N$. We next use this PMF to derive an ML-based estimator for the number of people $N$.
Given the inter-event times, we can estimate the number of people by maximizing the log-likelihood of the inter-event times. Specifically, assuming the inter-event times are independent, the log-likelihood of the observed inter-event times, $T_1, T_2, \dots ,T_n$, can be characterized as a function of the number of people $M$ as follows:
\begin{equation}\label{eq_loglikelihood}
LL(M) =  \sum_{i=1}^{n} \textnormal{log}(f_p(T_i;M)).
\end{equation}
We can then estimate the number of people by maximizing the log-likelihood function,
\begin{equation}\label{eq_ml_est}
\widehat{N}_\textnormal{renew} =  \argmax_M LL(M),
\end{equation}
where $\widehat{N}_\textnormal{renew}$ is the estimate of the number of people based on the underlying renewal-type process and the inter-event times. We note that we derived (\ref{eq_loglikelihood}) under the assumption that $T_i$'s are independent.  As we mentioned earlier, this is not necessarily the case for our process.  Thus, the ML estimator of (\ref{eq_ml_est}) is not the optimum, but can provide a good estimate of the number of people, as we shall see in the next section, while maintaining a low computational complexity.

In order to implement our derived estimator, one needs to identify the inter-event times due to the LOS blocking.  Furthermore, an estimate of $F_c(z)$, the CCDF of the inter-event times when single person is walking, is needed. In the next section, we show how the inter-event times and $F_c(z)$ can be estimated in practice.
\end{subsection}
\end{section}
\begin{section}{Experimental Results}{\label{sec_exp_results}}
In this section, we validate our proposed framework 
through extensive experiments. We start by explaining our experimental setup and then present the experimental results for five different areas with up to and including $20$ people.
\begin{figure}[!t]
\centering
\includegraphics[width=0.8\linewidth]{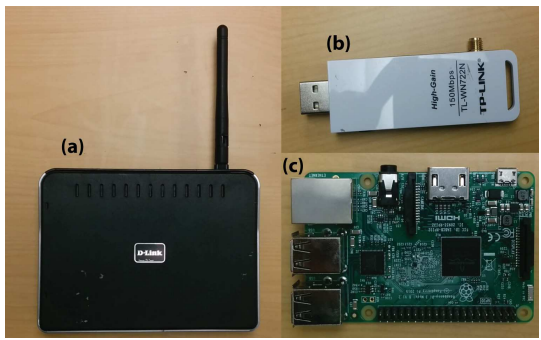}
\caption{(a) D-link WBR-1310 Router used as a WiFi Tx, (b) the WLAN card used as a WiFi Rx, and (c) Raspberry Pi board that controls the measurement operation and stores the WiFi RSSI measurements.}
\vspace{-0.28in}
\label{fig_hardware}
\end{figure}
\vspace{-0.05in}
\subsection{Experiment Setup}
As shown in Fig. \ref{fig_workspace}, our experimental setup consists of a pair of WiFi nodes for transmission and reception of wireless signals. One of the WiFi nodes is configured as a Tx, which constantly transmits wireless signals. The other WiFi node, which acts as a Rx, measures the signals that are emitted from the Tx node and records the corresponding signal strength. We use a D-Link WBR-1310 WiFi router \cite{wbr_router} as a Tx node, which operates using 802.11g wireless standard. For the Rx WiFi node, we use a TP-Link Wireless N150 WLAN card \cite{tplink_wlan_card} in 802.11g mode. This wireless card needs to be interfaced with a computer in order to make WiFi measurements. In our setup, we then use a Raspberry Pi (RPI) board \cite{rpi} for this purpose, i.e., to collect and store WiFi RSSI measurements. Fig. \ref{fig_hardware} shows the WiFi router, the WLAN card, and the RPI board used in our experiments. Omnidirectional antennas that come along with the WiFi router/card are used for transmitting and receiving the wireless signals. We use standard $2.4$ GHz frequency of WiFi in all our experiments.

Using the aforementioned experimental setup, we then run several experiments when up to $20$ people walk in the area of interest. We next first discuss the processing of the experimental data, which is followed by our experimental results.
\vspace{-0.15in}
\subsection{Initial Data Processing}\label{sec_data_processing}
In Section \ref{sec_estimation_method}, we developed a framework to estimate the number of people based on the PMF of the inter-event times, where an event refers to an instant of time where $l >0 $ people are crossing the LOS link. As discussed in Section \ref{sec_prob_statement}, the RSSI measurements are significantly attenuated when people cross the LOS link. Therefore, the RSSI measurements contain information about the times at which a cross has occurred and hence about the inter-event times. However, the received measurements are not only affected by the LOS blockage but also by the multipath fading that is caused by scattering off of the people that are not necessarily on the direct LOS. Therefore, we need to identify the times at which a LOS cross has occurred in the presence of multipath. 

 Our analysis of several measurements has shown that the fluctuations and dips caused by multipath are typically much smaller than those caused by any LOS blockage. Fig. \ref{fig_9_ppl_LOS_MP_effects} (left), for instance, shows the RSSI measurements of an experiment with $9$ people walking in an area, while Fig. \ref{fig_9_ppl_LOS_MP_effects} (right) shows the corresponding RSSI measurements in the same area but when the same number of people were instructed not to cross the LOS link. More specifically, $4$ people were instructed to walk on one side of the LOS link, with the other $5$ walking on the other side, without any person crossing the LOS link. Since there is no LOS blocking in this second case, the fluctuations in the RSSI measurements are solely due to the multipath effect. As can be seen, the measurements in Fig. \ref{fig_9_ppl_LOS_MP_effects} indicate that the effect of LOS blocking is more significant compared to the multipath effect. Specifically, the fluctuations in the RSSI measurements due to multipath are concentrated around the mean level of the RSSI signal, while blocking the LOS causes a pronounced dip in the signal level.
 \begin{figure}[t]
 \centering
 \includegraphics[width=1\linewidth]{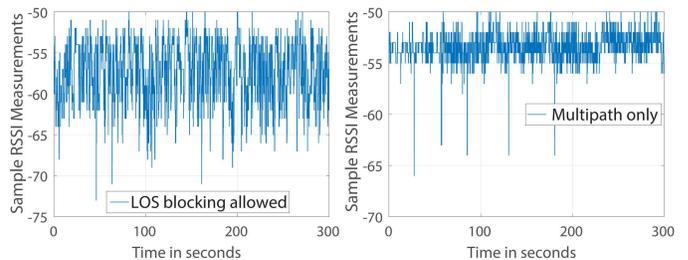}
 \vspace{-0.2in}
 \caption{(left) shows a sample RSSI power measurement when $9$ people are walking inside a building while (right) shows the RSSI power measurements in the same environment and for the same number of people when people are instructed not to cross the LOS link. The right figure thus mainly captures the fluctuations due to multipath fading. By comparing the two figures, it can be seen that the effect of LOS blocking is considerably more significant as compared to the fluctuations due to multipath.}\label{fig_9_ppl_LOS_MP_effects}
 \vspace{-0.24in}
 \end{figure}
\begin{figure*}[ht]
	\begin{minipage}{0.75\textwidth}
		\includegraphics[width=1\linewidth]{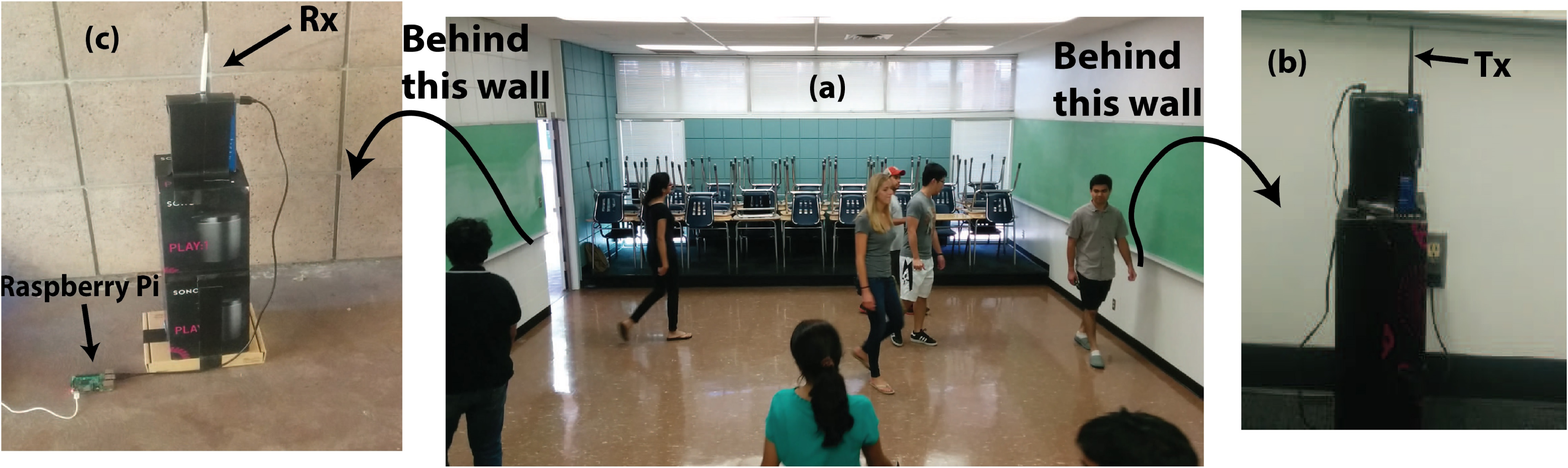}
		\captionof{figure}{(a) The first area of interest (Area 1), a closed classroom on our campus with wall made of concrete, where people are walking, (b) the Tx WiFi node located outside the classroom, behind one of the walls of the classroom as marked, and (c) the Rx WiFi node, along with the Raspberry Pi board that is used to control the data collection, which is located outside of the classroom behind the wall that is indicated.}
		\label{fig_classroom_scenario}
		\vspace{-0.25in}
	\end{minipage}
	\hspace{0.03in}
	\begin{minipage}{0.2\textwidth}
		\centering
		\resizebox{0.8\textwidth}{!}{%
			\begin{tabular}{|C{1cm}|C{1.5cm}|}
				\hline
				Number of People   & Estimated Number of People \\ \hline
				1 & 1  \\ \hline
				3 & 3  \\ \hline
				5 & 4  \\ \hline
				7 & 7  \\ \hline
				9 & 9  \\ \hline
		\end{tabular}}
		\captionof{table}{A sample result for counting through walls based on our proposed approach, for the classroom scenario of Fig. \ref{fig_classroom_scenario} on our campus (Area 1).}	\label{table_classroom}	
		\vspace{-0.1in}
	\end{minipage}
\end{figure*}
\begin{figure}[t]
	\centering
	\includegraphics[width=0.7\linewidth]{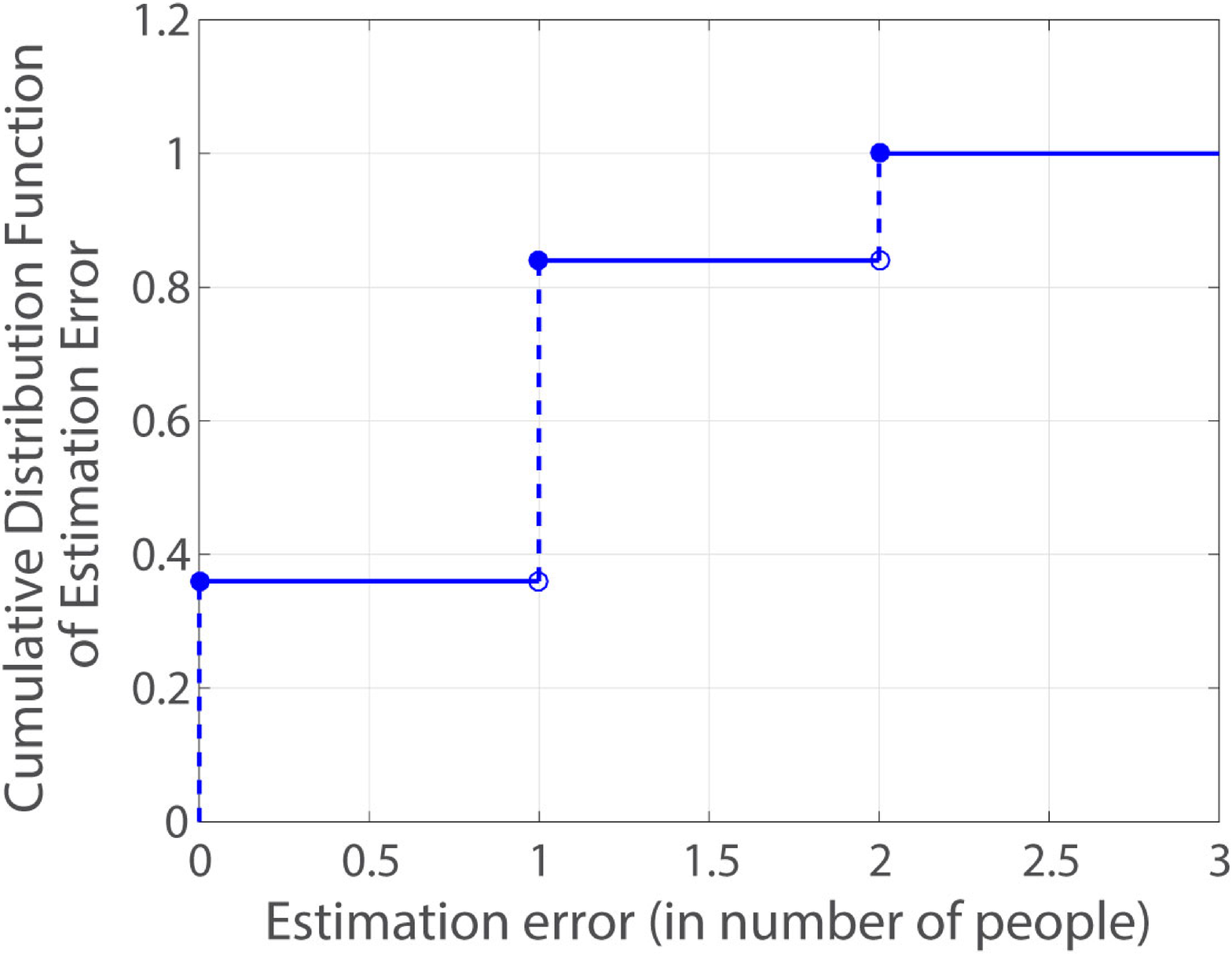}
	\vspace{-0.04in}
	\caption{The cumulative distribution function of counting estimation error based on $5$ sets of experiments in the classroom area of Fig. \ref{fig_classroom_scenario} on our campus (Area 1). In each set of experiment, we asked $1,\  3,\  5,\  7,\ $ and $9$ people to walk in the classroom. }
	\label{fig_cdf_plot}
	\vspace{-0.3in}
\end{figure}
 Based on several similar observations, we then contribute any dip in the RSSI signal level that is larger than a sufficiently-large threshold, $T_\textnormal{LOS}$, to people blocking the LOS link.\footnote{Note that we are only interested in detecting the time instants where any number of people block the LOS link, without the need to know the particular number of people that are along the LOS.} Furthermore, if $T_\textnormal{LOS}$ is chosen properly (not too large), then the chance of filtering a dip that was due to the LOS blockage becomes low. Thus, we utilize this approach in our experiments in order to identify the events of people crossing the LOS and hence the inter-event times.

In terms of the choice of the threshold, we choose the threshold $T_\textnormal{LOS}$ of $5$ dB in all our experiments, based on several observations similar to Fig. \ref{fig_9_ppl_LOS_MP_effects}. 
This means that any dip that is larger than $T_\textnormal{LOS}$ is labeled as a LOS blockage.
We note that, based on our observations, the choice of $T_\textnormal{LOS}$ is not strongly dependent on the area of interest, which allows us to set it without the need to make prior measurements in a specific area of interest. 
 We thus use the same value of $T_\textnormal{LOS}$ in all the five areas of interest considered in the next section. Furthermore, as we shall see in the sensitivity analysis of Section \ref{sec_sensitivity_analysis}, the threshold $T_\textnormal{LOS}$ is not sensitive to the specifics of the scenario such as the density of people and their walking speeds. Thus, $T_\textnormal{LOS}$ estimated with a specific number of people walking at a specific speed can be used to estimate a different number of people walking at other speeds, and more importantly in other areas. Finally, Section \ref{sec_sensitivity_analysis} explicitly shows that our experimental results are not that sensitive to the assumed $T_\textnormal{LOS}$ and moderate errors in estimating $T_\textnormal{LOS}$ are well tolerated. 
\subsection{Considering the Temporal Width of a Dip}
\vspace{-0.04in}
In practice, when a person crosses the LOS, the drop in the signal level is not an impulse drop. Rather, crossing the LOS link takes a finite amount of time, which means that each dip will have a small temporal duration.  Let $T_\textnormal{min}$ denote this time. Therefore, a person crossing the LOS link blocks the signal for a period of time $T_\textnormal{min}$. This then implies that we can not identify inter-event times that are less than $T_\textnormal{min}$.  In other words, any two events of crossing the LOS with an inter-event time smaller than $T_\textnormal{min}$ are not identifiable. Therefore, given that we can only identify inter-event times that are larger than $T_\textnormal{min}$ in practice, we modify our derived PMF of (\ref{eq_PMF_N_people}) to account for this. Then, $T_i,\  \forall\  i \in {1,2, \dots,n}$, is given as follows:
\begin{equation} \label{eq_modified_pmf}
\begin{split}
T_i|T_i \geq T_\textnormal{min} \sim   f_p^\textnormal{mod}(m;N) \triangleq \frac{f_p(m;N)}{{\displaystyle \sum_{r=T_{min}}^{\infty} f_p(r;N)}} 
\end{split}
\end{equation}
$f_p$ in (\ref{eq_loglikelihood}) is then replaced with $f_p^\textnormal{mod}$ to estimate the total number of people. 
\begin{figure*}[t]
	\begin{minipage}{0.75\textwidth}
		\includegraphics[width=1\linewidth]{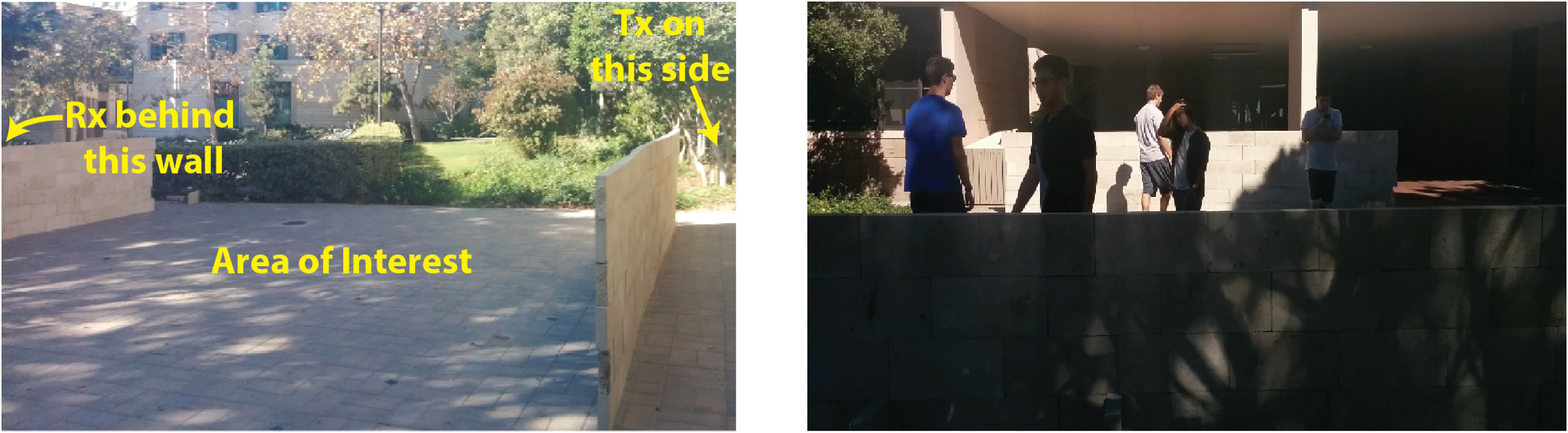}
		\vspace{-0.2in}
		\captionof{figure}{(left) The second area of interest between two concrete walls (Area 2), and (right) an example where people are walking in between these walls. The WiFi nodes are located outside of the area of interest, behind the walls, as indicated in the left figure\label{fig_hfh_wall}. Readers are referred to the color pdf for better visibility.}
		\vspace{-0.25in}
	\end{minipage}
	\hspace{0.03in}
	\begin{minipage}{0.2\textwidth}
		\centering
		\resizebox{0.8\textwidth}{!}{%
			\begin{tabular}{|C{1cm}|C{1.5cm}|}
				\hline
				Number of People   & Estimated Number of People \\ \hline
				1 & 3  \\ \hline
				3 & 5  \\ \hline
				5 & 6  \\ \hline
				7 & 6  \\ \hline
				9 & 7  \\ \hline
		\end{tabular}}
		\captionof{table}{A sample result for counting through walls based on our proposed approach, for the two-wall hallway scenario of Fig. \ref{fig_hfh_wall} on our campus (Area 2).}	\label{table_hfh}	
		\vspace{-0.22in}
	\end{minipage}
\end{figure*}
\begin{figure*}[t]
	\begin{minipage}{0.75\textwidth}
		\includegraphics[width=1\linewidth]{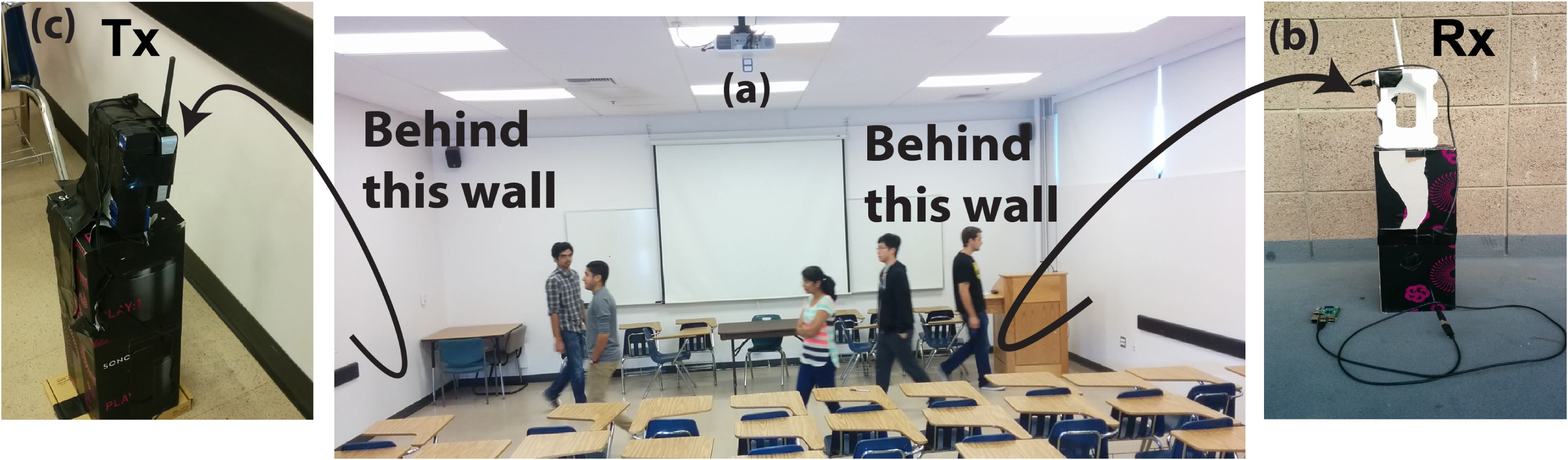}
		\vspace{-0.16in}
		\caption{(a) The third area of interest (Area 3), a closed classroom on our campus where people are walking. The room is enclosed by concrete walls on all four sides, (b) the Rx WiFi node located outside the classroom, behind one of the walls of the classroom as marked, and (c) the Tx WiFi node which is located outside of the classroom behind the wall that is indicated.} \label{fig_arts}
		\vspace{-0.28in}
	\end{minipage}
	\hspace{0.03in}
	\begin{minipage}{0.2\textwidth}
		\centering
		\resizebox{0.8\textwidth}{!}{%
			\begin{tabular}{|C{1cm}|C{1.5cm}|}
				\hline
				Number of People   & Estimated Number of People \\ \hline
				3 & 3  \\ \hline
				5 & 4  \\ \hline
				7 & 6  \\ \hline
				9 & 7  \\ \hline
		\end{tabular}}
		\captionof{table}{A sample result for counting through walls based on our proposed approach, for the classroom scenario of Fig. \ref{fig_arts} on our campus (Area 3).}	\label{table_arts}	
		\vspace{-0.35in}
	\end{minipage}
\end{figure*}
\begin{figure}[t]
	\centering
	\includegraphics[width=0.7\linewidth]{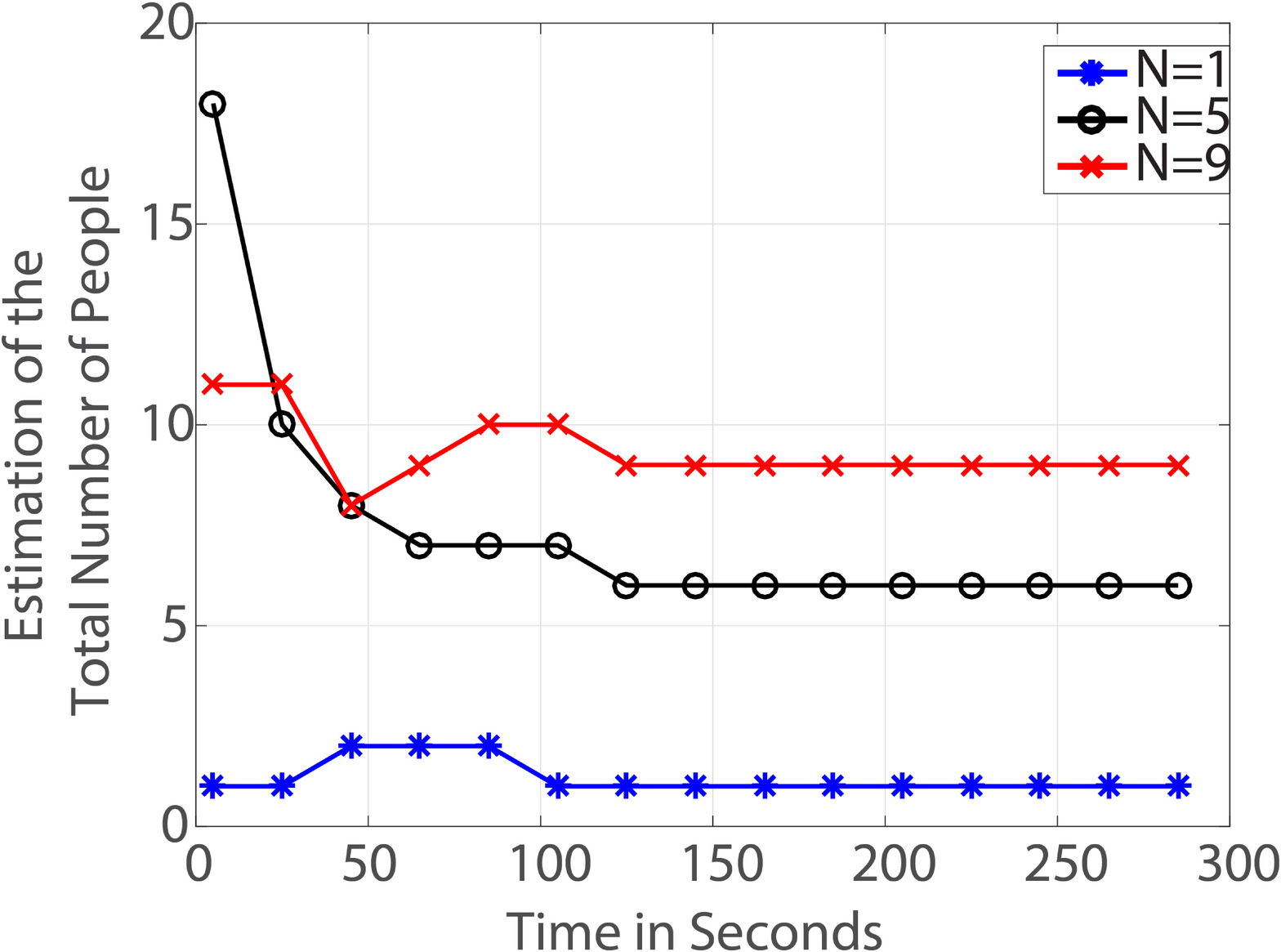}
	\vspace{-0.08in}
	\caption{Estimation of the total number of people as a function of time, for the classroom of Fig. \ref{fig_classroom_scenario} (Area 1) and for the three cases where $1$, $5$ and $9$ people are walking. It can be seen that the estimates converge to within one person of their final values within the first $100$ seconds.}
	\vspace{-0.3in}
	\label{fig_est_fun_time}
\end{figure}

The value of $T_\textnormal{min}$ depends on the speed of people. In this paper, we have assumed that people are walking casually. Based on simple experimental tests of one person crossing a link at a walking speed, we have chosen $T_\textnormal{min}= 1$ second in our results of the next section. We note that we do not need to measure this value in the particular experimental site of interest, as it is not that dependent on a particular site, but is rather more a function of the speed of people. Furthermore, as we shall see in section \ref{sec_sensitivity_analysis}, the experimental results are not that sensitive to the exact value of the assumed speed (and thus not that sensitive to the exact value of $T_\textnormal{min} $). 

The PMF of inter-event times in (\ref{eq_modified_pmf}) is a function of the CCDF of a single person inter-event times, $F_c(z)$, as shown in (\ref{eq_PMF_N_people}). In this paper, we obtain $F_c(z)$ using simulations. More specifically, we simulate motion of $1$ person using the motion model of Section \ref{sec_motion_model}. We then identify the times at which the person crosses the LOS link and extract the inter-event times. $F_c(z)$ is then obtained using these simulated inter-event times for a single person. We note that such a simulation is low in computation time (e.g., $1$ s), since it involves only one person.
\vspace{-0.2in}
\subsection{Experimental Results and Discussion}\label{sec_results_discussion}
To validate the proposed framework of Section \ref{sec_estimation_method}, we ran several experiments using the aforementioned experimental setup. We next present the results.

Fig. \ref{fig_classroom_scenario} \hl{shows the first experimental area (Area 1), which is a closed classroom on our campus, bounded by concrete walls on all four sides.} We asked people to walk inside the room while the WiFi nodes are located outside of the room, as shown in Fig. \ref{fig_classroom_scenario}. The walls are made of concrete bricks which are highly attenuating. The thickness of each wall is $20\ $cm based on our assessment. The dimensions of inside of the room, where people are walking, are $L=6.3 \ $m and $B=7.8\ $m, with the Tx and Rx positioned at $\frac{B}{2}$ (See Fig. \ref{fig_workspace}).

We have conducted several experiments in Area 1 when $1,3,5,7$, and $9$ people \hl{walked} inside the room. In each experiment, the measurements are collected for $300$ seconds at $20$ samples/sec. People are assumed to have a casual walking speed, which we take it to be $1\  $m/s in our theoretical modeling.\footnote{Note that we do not ask people to walk with a specific speed or in a specific pattern during the experiments. Instead, we simply ask them to walk casually in the area of interest.} Table \ref{table_classroom} shows sample results for the estimation of the number of people. It can be seen that our approach can estimate the total number of people walking inside the classroom with a high accuracy, by making WiFi measurements from outside, behind the classroom walls.  To further validate our framework statistically, we have run a series of experiments on different times/days to collect statistics of the estimation error. More specifically, we have run experiments on 5 different occasions in the classroom area of Fig. \ref{fig_classroom_scenario} (Area 1). In each run, $1,3,5,7$, and $9$ people are asked to walk in the classroom. Fig. \ref{fig_cdf_plot}  shows the Cumulative Distribution Function (CDF) of the estimation error based on these repeated measurements. It can be seen from the CDF plot that the estimation error is $1$ person or less $81$\% of the time and $2$ people or less $100$\% of the time, confirming a good statistical performance.

\begin{figure*}
	\begin{minipage}[t]{0.75\textwidth}
		\includegraphics[width=0.92\linewidth]{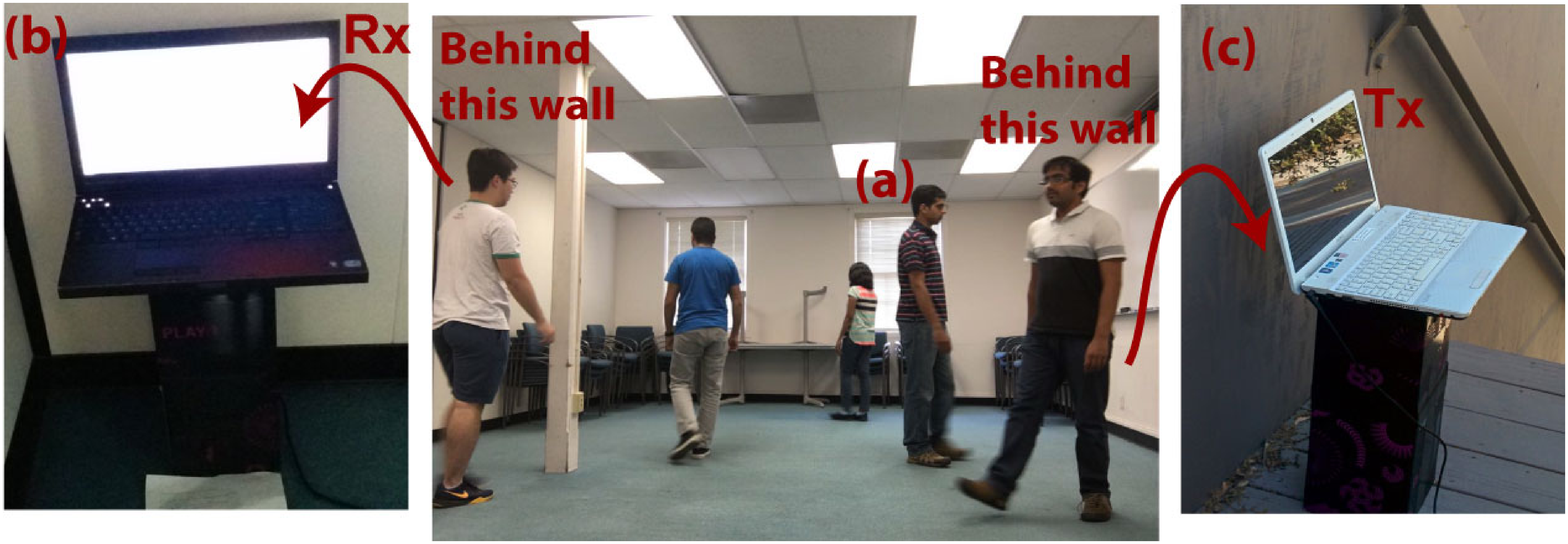}
		\vspace{-0.05in}
		\caption{(a) The fourth area of interest (Area 4), a closed conference room on our campus where people are walking. The room is enclosed by wooden walls on all four sides, (b) the Rx WiFi node located outside the room, behind one of the walls of the room as marked, and (c) the Tx WiFi node which is located outside of the room behind the wooden wall that is indicated. 
		}\label{fig_wooden_trailer}
	\end{minipage}
	\hspace{0.03in}
	\begin{minipage}{0.2\linewidth}
		\vspace{-0.9in}
		\resizebox{0.85\textwidth}{!}{%
		\begin{tabular}{|C{1cm}|C{1.7cm}|}
			\hline
			Number of People   & Estimated Number of People \\ \hline
			1 & 2  \\ \hline
			2 & 2  \\ \hline
			3 & 5  \\ \hline
			4 & 4 \\ \hline
			5 & 6  \\ \hline
			6 & 6  \\ \hline
			7 & 8  \\ \hline
			8 & 8  \\ \hline
			9 & 11  \\ \hline
		\end{tabular}}
		\captionof{table}{A sample result for counting through walls based on our proposed approach, for the classroom scenario of Fig. \ref{fig_wooden_trailer} on our campus (Area 4).}\label{table_woodenroom}	
	\end{minipage}
\vspace{-0.1in}
\end{figure*}
\begin{figure*}
	\begin{minipage}{0.7\textwidth}
		\centering
		\vspace{-0.35in}
		\includegraphics[width=1\linewidth]{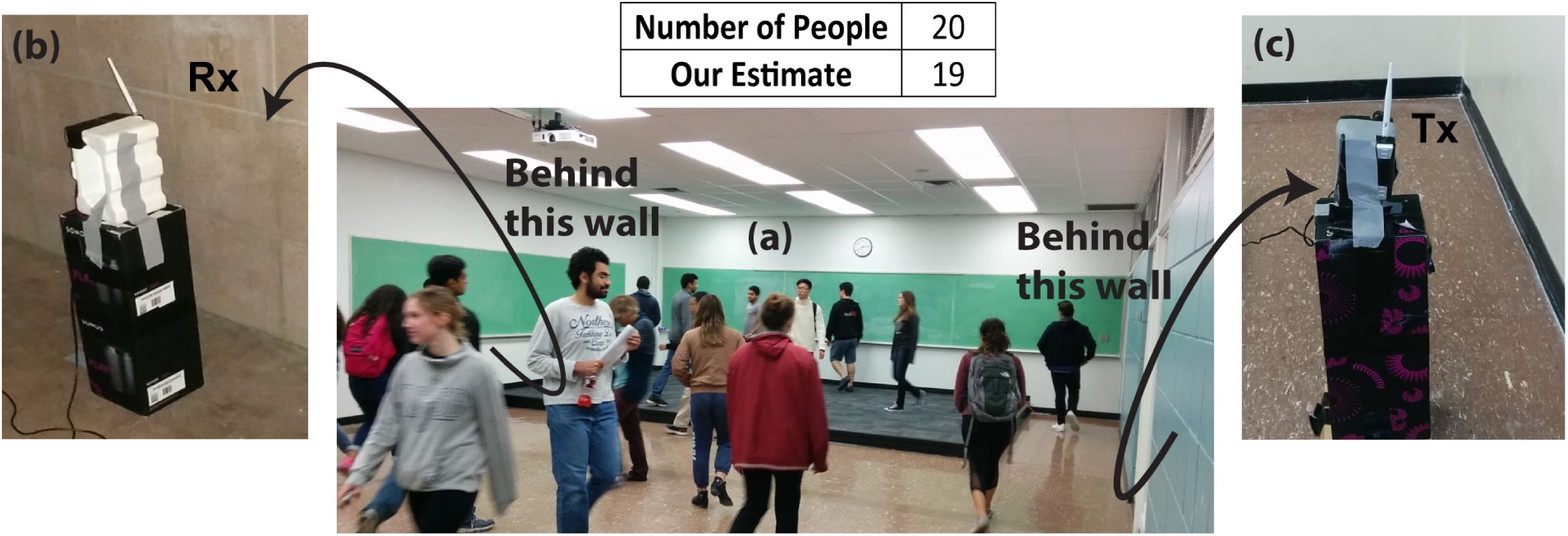}
		\vspace{-0.2in}
		\caption{(a) The fifth area of interest (Area 5), a closed classroom on our campus where people are walking. The room is enclosed on all four sides by walls that are made of a mixture of concrete and plaster, (b) the Rx WiFi node located outside the classroom, behind one of the walls of the classroom as marked, and (c) the Tx WiFi node which is located outside of the classroom behind the wall that is indicated. The performance of our framework with $20$ people walking in this area is also shown. It can be seen that our framework accurately estimates the number of people.}\label{fig_classroom_phelps_20ppl}
		\vspace{-0.1in}
	\end{minipage}
	\hspace{0.12in}
	\begin{minipage}{0.25\textwidth}
		\centering
		\resizebox{0.7\textwidth}{!}{%
		\begin{tabular}{|C{1cm}|C{1.5cm}|}
			\hline
			Number of People   & Estimated Threshold (dB) \\ \hline
			1 & 4  \\ \hline
			3 & 4  \\ \hline
			5 & 4  \\ \hline
			9 & 5  \\ \hline
		\end{tabular}}
		\captionof{table}{Sensitivity of the estimated threshold $T_\textnormal{LOS}$ to the number of people walking in the area. It can be seen that the optimum threshold is not that sensitive to the number of people in the area.}	\label{table_sensitivity_thr_num_ppl}
	\end{minipage}
\vspace{-0.2in}
\end{figure*}

To further validate our approach, we next run experiments in an outdoor area  occluded by walls. Fig. \ref{fig_hfh_wall} shows the outdoor area of interest (Area 2). As can be seen, two parallel walls are constructed with concrete bricks. The thickness of each wall is $5\  $cm in this case. The dimensions of the area of interest are $ L=10\  $m and $B=7\ $m. People are then asked to walk in the hallway created in between the walls, while a Tx and a Rx node are mounted outside of each wall.  Table \ref{table_hfh} shows a sample result obtained in Area 2. As can be seen, the number of people is estimated with a good accuracy. Fig. \ref{fig_arts} shows a third area of interest (Area 3), which is another classroom on our campus. The area is bounded by concrete walls on all four sides. People walk in part of this room with the dimensions of $ L=7.8\  $m and $B=3.96\ $m as shown in Fig. \ref{fig_arts}. Note that Area 3 has rich multipath due to the furniture in the room. Table \ref{table_arts} shows a sample result obtained in Area 3. It can be seen that the number of people are estimated accurately.

To further validate our framework with walls made of different material than concrete, we ran experiments in a room enclosed by wooden walls. Fig. \ref{fig_wooden_trailer} shows the fourth area of interest (Area 4), which is a conference room on our campus. The dimensions of the area of interest are $L=4.1\  $m and $B=7.5\ $m. We then run experiments with up to and including $9$ people in this area. Table \ref{table_woodenroom} shows the performance of our framework in this case. It can be seen that our framework can estimate the number of people with a high accuracy, which shows the robustness of our approach to the wall material.

So far, we demonstrated experimental results with up to and including $9$ people in $4$ different areas on our campus. To test the scalability of our approach, we further run experiments with $20$ people walking inside a classroom. Fig. \ref{fig_classroom_phelps_20ppl} shows the fifth area of interest which is a classroom on our campus enclosed on all four sides by walls that are made of a mixture of concrete and plaster (Area 5). The dimensions of this area are $ L=7.9\  $m and $B=12.6\ $m. We then run experiments with $20$ people walking inside this classroom as shown in Fig. \ref{fig_classroom_phelps_20ppl}. Our framework estimates the number of people inside as $19$ in this case, which shows the scalability of our framework to the higher number of people with only one WiFi link.\footnote{We note that as the size of the area and the number of people increases, at some point we inevitably have to use more links. However, the fact that $20$ people can be counted through walls with only one WiFi link in an area of the size $100\ \textnormal{m}^2$ is promising for how this approach will scale to bigger areas and more people.}
This experiment further tests the proposed approach with a third kind of wall material, a mixture of plaster and concrete, and confirms its robustness. Overall, considering all five areas, our framework can estimate up to and including $20$ people with an error of $2$ people or less $100$\% of the time and with an error of $1$ person or less $75$\% of the time. 

So far, we have demonstrated that the proposed framework can estimate the total number of people walking inside an occluded area of interest with a high accuracy. In all the experimental results so far, we have used data collected for $300$ seconds. Next, we show the time we need to wait before the estimates converge to their final values. More specifically, Fig. \ref{fig_est_fun_time} shows the estimates as a function of time for an experiment with $1$, $5$, and $9$ people for the classroom scenario of Fig. \ref{fig_classroom_scenario} (Area 1). It can be seen that the estimates converge to within $1$ person of their final values within the first $100$ seconds.
\begin{figure*}[t]
	\begin{minipage}{0.4\textwidth}
		\centering
		\includegraphics[width=0.73\linewidth]{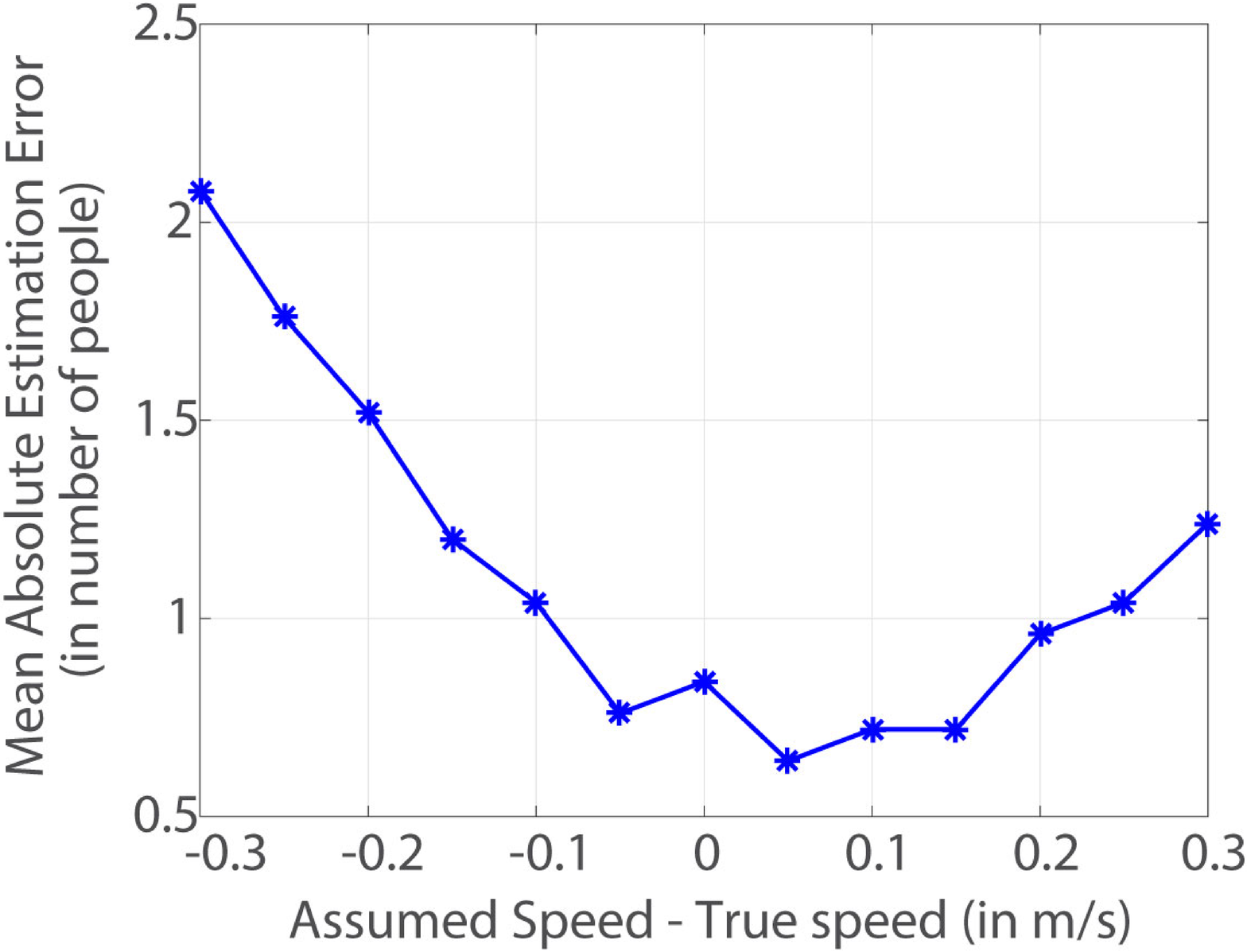}
		\vspace{-0.15in}
		\caption{Sensitivity of our crowd counting to the errors in the assumed walking speed. The casual speed of people is assumed to be 1 m/s.  Other speeds were then assumed in our derivations when estimating the number of people.  It can be seen that our framework is robust to moderate errors in the assumed speed.}
		\label{fig_sensitivity_vel}
	\end{minipage}
	\hspace{0.15in}
	\begin{minipage}{0.3\textwidth}
		\vspace{-0.38in}
		\includegraphics[width=1.08\linewidth]{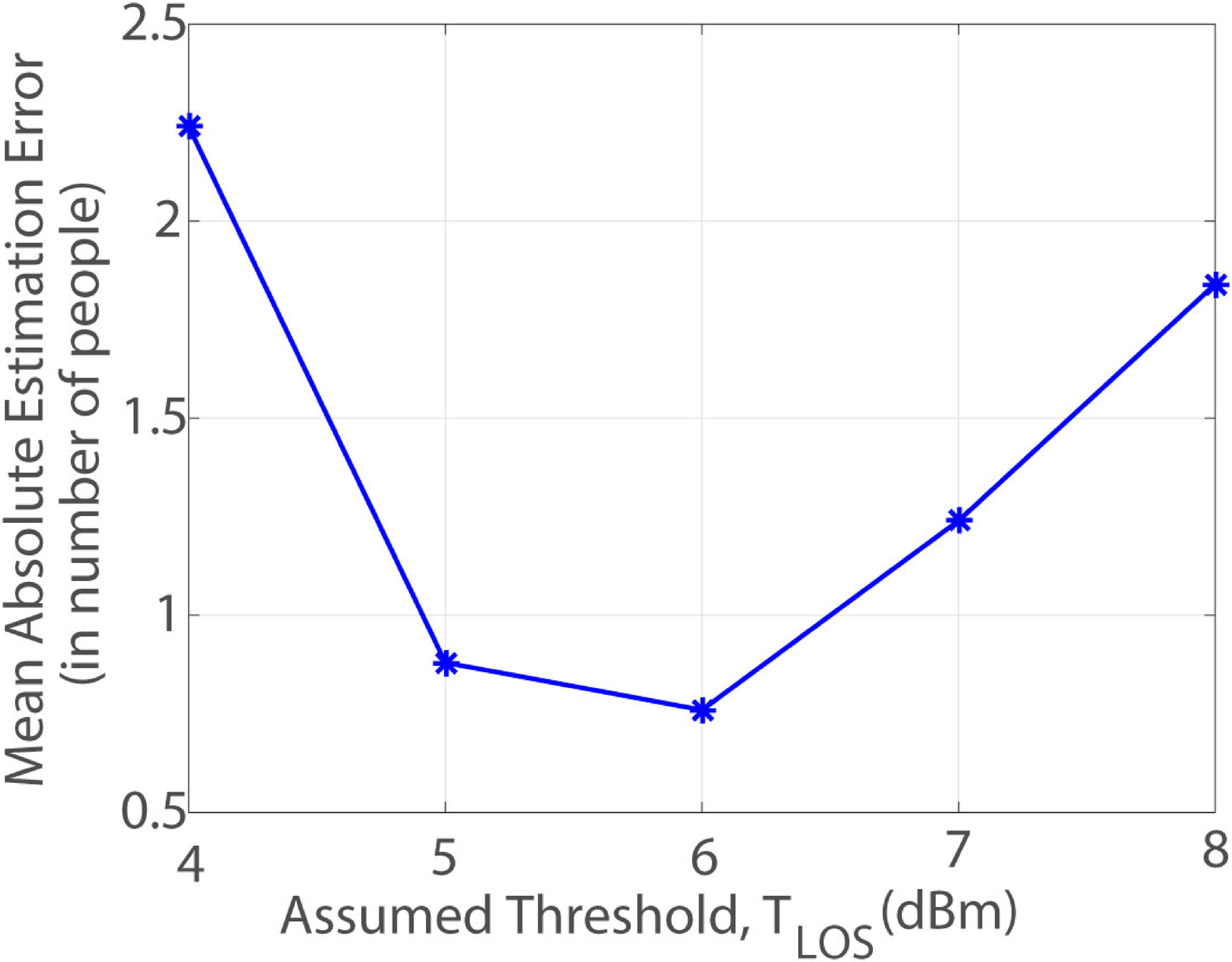}
		\vspace{-0.2in}
		\caption{Sensitivity of our crowd counting to the errors in the assumed threshold $T_\textnormal{LOS}$. It can be seen that our framework is robust to moderate errors in the assumed value of the threshold. }
		\vspace{-0.3in}
		\label{fig_sensitivity_threshold}
	\end{minipage}
	\hspace{0.15in}
	\begin{minipage}{0.2\textwidth}
		\resizebox{0.9\textwidth}{!}{%
			\begin{tabular}{|C{1.5cm}|C{1.5cm}|}
				\hline
				Speed of People   & Estimated Threshold (dB) \\ \hline
				Standing Still & 4  \\ \hline
				Normal Walking & 5  \\ \hline
				Running & 4  \\ \hline
		\end{tabular}}
		\captionof{table}{Sensitivity of the optimum threshold $T_\textnormal{LOS}$ to the speed of people walking in the area. It can be seen that the estimated optimum threshold is not that sensitive to the speed of people in the area.}	\label{table_sensitivity_thr_speed_ppl}	
	\end{minipage}
	\vspace{-0.28in}
\end{figure*}

Overall, our experimental results confirm that the proposed \hl{framework can} estimate the number of people inside a room or a building, or in general behind walls, solely from WiFi RSSI measurements acquired from outside, with a good accuracy.
\subsection{Sensitivity Analysis}\label{sec_sensitivity_analysis}
In the experimental results of this section, we took $T_\textnormal{LOS}$ as $5$ dB and assume a walking speed of $1$ m/s. We next show that the our framework is not sensitive to the exact value of $T_\textnormal{LOS}$ and the speed of people and that moderate errors in both can be well tolerated. Furthermore, we show that the estimation of $T_\textnormal{LOS}$ is not that sensitive to the specifics of the scenarios such as the density of people and their walking speeds. This then greatly reduces the calibration demand of our approach as $T_\textnormal{LOS}$ estimated with a specific number of people walking at a specific speed in the calibration phase can be used to estimate a different number of people walking at a different speed in the estimation phase. We furthermore have only calibrated $T_\textnormal{LOS}$ in one area and have used it in the other $4$ areas in all our experiments, which indicates the generalizability of it across different areas, further reducing the calibration burden, which is important for behind-wall scenarios.
\subsubsection{Sensitivity to the Assumed Walking Speed}\label{sec_sens_walk_speed}
The results of Section \ref{sec_results_discussion} assumed that the people in the area of interest are walking at an average speed of \hl{$1$ m/s, based on the typical} walking speed of humans. However, the average walking speed could vary slightly from this value depending on the person or the environment, for instance due to the density of people in the region. In this section, we consider the effect of errors in the assumed walking speed (as compared to the true speed of people) on the estimation of the number of people.

In order to analyze the effect of the assumed walking speed and its deviation from the true speed of people during the experiment, consider an experiment where people are told to walk casually, which amounts to a speed of around 1 m/s. We then assume that people are walking at a speed of $v$ m/s in our derivations and estimate the number of people based on our framework. Fig. \ref{fig_sensitivity_vel} shows the mean absolute estimation error in the number of people as a function of the error in the assumed walking speed. At each assumed speed, the estimates with different number of people ($N$=1, 3, 5, 7, and 9) walking in the classroom area of Fig. \ref{fig_classroom_scenario} are obtained over 5 repeated sets of experiments. The mean of the absolute error is then shown in Fig. \ref{fig_sensitivity_vel}. It can be seen that the estimation error is less than 2 people in most of the assumed speed range, showing the robust nature of our framework to small errors in the assumed walking speed of the people as compared to the true speed.
\subsubsection{Sensitivity to the Assumed Threshold}
As explained in Section \ref{sec_data_processing}, a threshold $T_\textnormal{LOS}$ is used to separate the dips of the wireless measurements that are due to people blocking the LOS path from the dips due to multipath. The time instants at which these dips occur are then used to estimate the number of people in the area as explained in Section \ref{sec_estimation_method}. As discussed in Section \ref{sec_data_processing}, we have used $T_\textnormal{LOS}=5$ dB in all our results. However, the true optimal value of $T_\textnormal{LOS}$ is hard to quantify. In this section, we consider the impact of the choice of $T_\textnormal{LOS}$ on the estimates of the number of people. More specifically, we consider a range of values for $T_\textnormal{LOS}$ and estimate the number of people. At each $T_\textnormal{LOS} $, the estimates with different number of people ($N$=1, 3, 5, 7, and 9) walking in the classroom area of Fig. \ref{fig_classroom_scenario} are obtained over 5 repeated sets of experiments. The mean of the absolute error is then shown in Fig. \ref{fig_sensitivity_threshold}. As can be seen, the mean error is less than $2$ people for a wide range of $T_\textnormal{LOS}$, which shows the robust nature of our framework to moderate errors in the estimated threshold $T_\textnormal{LOS}$. 
\subsubsection{Sensitivity of the Threshold to the Density and Speed of People}
As explained in Section \ref{sec_data_processing}, the threshold $T_\textnormal{LOS}$ is estimated by collecting wireless measurements when people are walking without blocking the LOS link. This threshold is then used to separate the LOS blockage from the multipath. For instance, the estimate of $T_\textnormal{LOS} = 5$ dB used in all our experiments is obtained in the calibration phase when $9$ people are walking on either side of the LOS link in one area. In this section, we consider the effect of the number of people walking in the area and their walking speed in estimating $T_\textnormal{LOS}$. More specifically, we let different number of people ($N$=1, 3, 5, and 9) walk on either side of the LOS link without blocking the LOS link as explained in Section \ref{sec_data_processing}. Furthermore, we let $9$ people walk at three different speeds of standing still, normal walking, and running. Table \ref{table_sensitivity_thr_num_ppl} and \ref{table_sensitivity_thr_speed_ppl}  show the estimated threshold as a function of the number of people in the area and their walking speeds, respectively. It can be seen that the estimated threshold is not that sensitive to the number of people walking in the area or to their speeds, which explains the good accuracy of our results with different number of people and with an assumed speed of $1$ m/s.
\vspace{-0.1in}

\end{section}
\vspace{-0.05in}
\begin{section}{Conclusions}{\label{sec_conclusions}}
\vspace{-0.1in}
In this paper, we proposed a framework to count the total number of people walking in an area that is occluded by walls, using only the RSSI of WiFi transceivers that are installed outside of the area. We proposed to use the inter-event times corresponding to the signal dips for crowd counting through walls as it is more robust to the attenuation of the walls. More specifically, we showed how to model the impact of people on the received power measurements using superposition of Renewal-type processes. We then mathematically characterized the statistics of the inter-event times of the resulting process and showed how it contains vital information on the total number of people, which then became the base for our ML estimation of the total number of people. To validate our proposed framework, we ran extensive experiments in five different areas on our campus, three classrooms, a conference room, and a hallway, with walls made of different material such as concrete, plaster, and wood, and with up to and including $20$ people, and showed that our approach can estimate the total number of people through walls with a high accuracy.
\end{section}
\bibliographystyle{IEEEtran}
\bibliography{ref_cctw}

\end{document}